\documentclass{article}
\usepackage{wright}

\newcommand{\errors}{\calE}
\newcommand{\errorchan}{\calN}
\newcommand{\channels}[1]{\mathsf{Channels}(#1)}
\newcommand{\Enc}{\mathsf{Enc}}
\newcommand{\Dec}{\mathsf{Dec}}
\newcommand{\Span}{\mathrm{span}}

\title{Haar random codes attain the quantum Hamming bound, approximately}
\author{Fermi Ma\thanks{Simons Institute $\&$ UC Berkeley. Email: \texttt{fermima1@gmail.com}.} \and Xinyu Tan\thanks{MIT. Email: \texttt{norahtan@mit.edu}.} \and John Wright\thanks{UC Berkeley. Email: \texttt{jswright@berkeley.edu}.}}
\date{}

\begin{document}

\maketitle

\begin{abstract}
We study the error correcting properties of Haar random codes,
in which a $K$-dimensional code space $\bC \subseteq \C^N$ is chosen at random from the Haar distribution.
Our main result is that Haar random codes can approximately correct errors up to the quantum Hamming bound,
meaning that a set of $m$ Pauli errors can be approximately corrected so long as $mK \ll N$.
This is the strongest bound known for any family of quantum error correcting codes (QECs),
and continues a line of work showing that approximate QECs can significantly outperform exact QECs~\cite{LNCY97,CGS05,BGG24}.
Our proof relies on a recent matrix concentration result of Bandeira, Boedihardjo, and van Handel~\cite{BBV23}.
\end{abstract}

\hypersetup{linktocpage}
\tableofcontents

\section{Introduction}

One of the foundational achievements of quantum computing is the development of quantum error-correcting codes (QECs).
These are the basic building blocks of fault tolerance and a necessary ingredient in the development of scalable quantum computers,
but they have also found applications to areas as diverse as quantum complexity theory and many-body physics.
In 1997,
Leung, Nielsen, Chuang, and Yamamoto~\cite{LNCY97} introduced a relaxation of QECs known as an \emph{approximate} quantum error-correcting codes (AQECs), 
in which a codeword, after being subjected to noise,
is only required to be recovered \emph{approximately} rather than exactly.
In their work, they proved the remarkable fact that AQECs can sometimes correct more errors than exact QECs are capable of.
Follow-up works of Cr\'{e}peau, Gottesman, and Smith~\cite{CGS05} and Bergamaschi, Golowich, and Gunn~\cite{BGG24} proved strong quantitative versions of this fact, which state that AQECs can achieve parameters well beyond those ruled out for exact QECs by the quantum Singleton bound.
The latter of these works, in particular, gives a construction of an AQEC which encodes $k$ qudits into $n$ qudits and can correct errors on roughly $(n-k)/2$ of the qudits, which is \emph{twice} what is possible with an exact QEC due to the quantum Singleton bound.

We study the error correcting properties of \emph{Haar random codes}.
In these codes, message states $\ket{m} \in \C^K$ are encoded with a Haar random isometry $\bV: \C^K \mapsto \C^N$ to produce code states $\ket{\bpsi} = \bV \cdot \ket{m} \in \C^N$.
Due to the immense amount of randomness needed to generate the encoding isometry, Haar random codes are essentially without structure, and as a result they can only be error corrected approximately rather than exactly.
But this lack of structure has benefits as well,
because
it should be difficult for an error to corrupt one codeword into another,
and so Haar random codes should be extremely resilient to noise.
We formalize this intuition by proving that Haar random codes error correct as well as one could possibly hope for.
In particular, we show that given a linear space of errors $\errors = \mathrm{span}\{E_1, ..., E_m\}$ generated by $m$ fixed Pauli matrices,
if a noise channel $\errorchan$ whose errors come from $\errors$ is applied to a code state $\ket{\bpsi}$,
then one can recover $\ket{\bpsi}$ up to negligible error in trace distance so long as $Km \ll N$.
Our results hold for all ``sufficiently large'' codespaces, which satisfy the minor requirement that the codespace dimension $K$ is at least $(\log_2 N)^3$.

To understand our bound, let us recall the well-known \emph{quantum Hamming bound} on the parameters of a QEC.
It pertains specifically to \emph{nondegenerate} QECs, which are exact quantum codes with the following property: if $\ket{v_1}, \ldots, \ket{v_K}$ is an orthonormal basis of the codespace $C \subseteq \C^N$, then the vectors $E_i \cdot \ket{v_j}$, for $i \in [m]$ and $j \in [K]$, are orthogonal to each other.
Since these vectors live in $\C^N$, this is of course only possible if $Km \leq N$,
and so the quantum Hamming bound, in its full generality, states that a $K$-dimensional nondegenerate code $C \subseteq \C^N$ correcting the Pauli errors $\{E_1, \ldots, E_m\}$ must satisfy $Km \leq N$.
This argument crucially relies on $C$ being nondegenerate, and it remains an important open problem to determine whether degenerate QECs also satisfy the quantum Hamming bound.

Our main result, then, stated more succinctly, is that Haar random codes approximately saturate the quantum Hamming bound.
This gives the first construction of an error correcting code which is known to do so for almost all parameters; in fact, it is \emph{impossible} for exact QECs to even approximately satisfy the quantum Hamming bound for most settings of parameters, as there are large ranges of parameters where the quantum Singleton bound is strictly stronger than the quantum Hamming bound.
Our proof shows that if $\ket{\bv_1}, ..., \ket{\bv_K}$ is a basis of the codespace $\bC$ of a Haar random code, then
the randomness in $\bC$ scatters the $Km$ vectors $E_i \ket{\bv_j}$ throughout the Hilbert space so effectively that they essentially act as an orthonormal basis;
furthermore, the only barrier to this holding is when the dimension $N$ is simply too small to fit $Km$ orthonormal vectors completely.
In a sense, this means that a Haar random code forms a type of \emph{approximate} nondegenerate code.
Thus, it is subject to the quantum Hamming bound,
and we show that it also attains it.

Our work is situated in a long tradition of research in the area of coding theory which studies properties of random ensembles of codes.
Although these ensembles of random codes cannot always be implemented efficiently,
there are many settings in which they are able to achieve optimal parameters, and so studying them helps us understand the ultimate limits of error correcting codes;
in addition, even when they do not achieve optimal parameters, they still serve as important benchmarks against which other codes may be compared.
In the literature on classical error correction, this typically means studying either random codes (i.e.\ uniformly random subsets of $[d]^n$, where $[d]$ is the alphabet of the code)
or random linear codes.
For example, dating back to 1957, Varshamov~\cite{Var57} showed that random linear codes attain the Gilbert-Varshamov (GV) bound,
which is one of the best known lower bounds on the parameters a classical code can achieve.
Going back even earlier, Shannon~\cite{Sha48} used random codes to prove the Shannon capacity theorem,
showing that they achieve channel capacity.

In quantum error correction, the two most commonly studied ensembles of random codes are random stabilizer codes and random CSS codes. They are often used to show the existence of ``good'' codes, i.e.\ codes with linear dimension and distance,
and they also provide examples of codes which attain quantum analogues  of the GV bound.\footnote{Confusingly, there are actually \emph{two} bounds in the literature which are commonly referred to as ``the quantum GV bound'', a weaker one attained by random CSS codes~\cite{Sho18} and a stronger one attained by random stabilizer codes~\cite{Smi06}.} 
Random stabilizer codes and random CSS codes can both be seen as quantum analogues of classical random \emph{linear} codes: whereas classical random linear codes can be specified by a list of random standard basis parity checks, these quantum codes can be specified by a list of random Pauli basis parity checks. However, far less seems to be known about Haar random codes, arguably the simplest and most natural family of random quantum codes. These codes can be seen as the quantum analogue of classical random codes: just as a classical random code encodes each  message $m$ as a uniformly random (distinct) classical string, a Haar random code encodes each message $\ket{m}$ as a uniformly random (orthogonal) quantum state. Despite their simplicity, even basic quantities such as the number of errors that Haar random codes can correct had not been pinned down prior to our work.

Haar random codes are also motivated by physics,
where they arise naturally in the study of highly disordered and chaotic systems such as black holes.
Given a $k$-qubit input state $\ket{m}$, a chaotic system can be modeled as appending some already-existing ancilla qubits $\ket{0^{n-k}}$ to $\ket{m}$ and then applying an $n$-qubit Haar random unitary $\bU$ to the entire system.
The resulting state $\ket{\bpsi} \coloneqq \bU \ket{m} \ket{0^{n-k}}$ is a code state in the Haar random code specified by the encoding isometry $\bV \coloneqq \bU\cdot (I \otimes \ket{0^{n-k}})$, and the question one wants to answer about the chaotic system often reduces to a question about the error correcting properties of this code.
A prominent example which follows this template is the Hayden–Preskill thought experiment~\cite{HP07},
which seeks to understand how well a state $\ket{m}$ which has fallen into a black hole can be recovered from the black hole's Hawking radiation.
Underlying their thought experiment in the case of a young black hole is a precise technical statement about the error correcting properties of Haar random codes, namely the fact that the original state $\ket{m}$ can be recovered to high fidelity, even if $\ell \ll (n-k)/2$ qubits have been erased from $\ket{\bpsi}$.
This fact actually follows as a very special case from our main result: an $\ell$ qubit erasure error can be modeled as a noise channel $\calN$ involving $m = 4^{\ell} \ll N/K$ Pauli matrices, and our main result implies that one can still recover $\ket{m}$ after $\errorchan$ has been applied, as $Km \ll N$.
However, our main result applies not just to the special case of erasure errors but to general Pauli errors, and it can be viewed as showing that chaotic systems have much stronger error correcting properties than previously known.

Our work leaves open the question of whether Haar random codes are optimal among all AQECs, or whether there exist codes which outperform them. 
Such a code would have to make fundamental use of degeneracy, as otherwise it would be subject to the same quantum Hamming bound that Haar random codes achieve.
As we have pointed out, it  remains open even whether \emph{exact} QECs must satisfy the quantum Hamming bound, or whether they can use degeneracy to surpass it.
We conjecture that Haar random codes are optimal, and that degeneracy cannot be used to outperform them.

\subsection{Approximate quantum error correction}

In order to state our results formally,
we will first define approximate quantum error correction.

\begin{definition}[Errors]
An \emph{error} on a Hilbert space $\calH$
is a linear operator $E$ which acts on $\calH$.
We will typically write $\errors$ for a set of errors.

Given a set of errors $\errors$, we write $\channels{\errors}$ for the set of quantum channels whose Kraus operators are contained in $\errors$. Formally, if $\errorchan$ is a quantum channel with Kraus decomposition
\begin{equation*}
    \errorchan(\rho) = \sum_{i} K_i \cdot \rho \cdot K_i^{\dagger},
\end{equation*}
then $\errorchan \in \channels{\errors}$ if $K_i \in \calE$ for all $i$.
\end{definition}

\begin{definition}[Approximate quantum error correcting code]\label{def:aqec}
    An \emph{approximate quantum error-correcting code (AQEC)} is specified by an isometry $V_{\Enc} : \C^K \rightarrow \C^N$ known as the \emph{encoding isometry}.
    Given a \emph{message state} $\ket{m} \in \C^K$,
    the corresponding \emph{encoded state} is $V_{\Enc} \cdot \ket{m}$.
    The subspace of all encoded states is given by the image of $V_{\Enc}$ and is known as the \emph{codespace}.
    We denote the encoding channel as $\Enc: \rho \mapsto V_\Enc \cdot \rho \cdot V_\Enc^\dagger$. 

    The code \emph{corrects for a set of errors $\calE$ with disturbance $\epsilon$} if there exists a \emph{decoding channel} $\Dec$ such that the following is true. For every noise channel $\errorchan \in \channels{\calE}$,
    \begin{equation*}
        \frac{1}{2} \norm{\Dec\circ \errorchan \circ \Enc - \calI}_{\diamond} \leq \epsilon, 
    \end{equation*}
    where $\calI$ denotes the identity channel. Note that the maximum probability an algorithm, given either $\Dec \circ \errorchan \circ \Enc$ or $\calI$ at random, can guess which channel it is given, is at most $1/2+ \varepsilon/2$, so $\varepsilon \leq 1$.
\end{definition}

We now define Haar random codes, the main topic of this work.

\begin{definition}[Haar measure]
    Write $U(N)$ for the group of $N \times N$ unitary matrices.
    The \emph{Haar measure} is the unique measure on $U(N)$ such that if $\bU$ is a random unitary distributed according to the Haar measure, then $\bU \cdot W$ and $W \cdot \bU$ are also distributed according to the Haar measure, where $W \in U(N)$ is any fixed unitary.
    For each $1 \leq i \leq N$, write $\ket{\bu_i}$ for the $i$-th column of $\bU$, so that we can write
    \begin{equation*}
    \bU = \sum_{i=1}^N \ketbra{\bu_i}{i}.
    \end{equation*}
\end{definition}
\begin{definition}[Haar random codes]
    A \emph{Haar random isometry} $\bV : \C^K \rightarrow \C^N$ is given by
    \begin{equation*}
        \bV = \sum_{i=1}^K \ketbra{\bu_i}{i},
    \end{equation*}
    where $\bU = \sum_{i=1}^N \ketbra{\bu_i}{i}$ is a Haar random unitary in $U(N)$. 
    This isometry maps $\C^K$ into the subspace of $\C^N$ given by $\bC = \Span\{\ket{\bu_1}, \ldots, \ket{\bu_K}\}$.
    A \emph{Haar random code of dimension $K$ in $\C^N$} is the code specified by a Haar random isometry $\bV$ of this form, and its associated codespace is given by $\bC$.
\end{definition}

\subsection{Our contributions}

\subsubsection{Unitary error sets}
We begin by introducing the class of errors that we study in this work,
which is significantly broader than the class of errors traditionally studied in quantum error correction.

\begin{definition}[Unitary error sets]
A set of errors $\{E_1, \ldots, E_m\}$ forms a \emph{unitary error set} if each $E_i$ is a unitary matrix, and for all $1 \leq i \neq j \leq m$, $\mathrm{tr}(E_i^\dagger E_j) = 0$.
Given a unitary error set $\{E_1, \ldots, E_m\}$, we will focus on correcting errors drawn from the set $\errors = \Span\{E_1, \ldots, E_m\}$. 
\end{definition}

As a special case of unitary error sets,
we can recover the Pauli error sets which are typically used to model noise in quantum error correction.

\begin{example}[Pauli matrices]
The single-qubit \emph{Pauli matrices} are the four matrices
\begin{equation*}
    I = \begin{pmatrix}
        1 & 0 \\
        0 & 1
        \end{pmatrix}, \quad
    X = \begin{pmatrix}
        0 & 1 \\
        1 & 0
        \end{pmatrix}, \quad
    Y = \begin{pmatrix}
        0 & -i \\
        i & 0
        \end{pmatrix}, \quad
    Z = \begin{pmatrix}
        1 & 0 \\
        0 & -1
        \end{pmatrix}.
\end{equation*}
The Pauli matrices are unitary and satisfy $\tr(P^{\dagger} Q) = 0$ for $P \neq Q \in \{I, X, Y, Z\}$ and therefore form a unitary error set of size 4.
They can be generalized to act on $n$ qubits
by considering the $n$-qubit Pauli matrices $\{I, X, Y, Z\}^{\otimes n}$.
These form a basis on the set of $n$-qubit matrices, and any subset of them $\{P_1, \ldots, P_m\}$ forms a unitary error set.
Two such subsets will be especially important for us.
\begin{itemize}
    \item[$\circ$] Let $\{E_1, \ldots, E_m\}$ contain all Paulis which act on the first $t$ qubits, so that $m = 4^t$, and set $\calE = \Span\{E_1, \ldots, E_m\}$.
    Then $\calE$ contains all matrices of the form $M \otimes I_{n-t}$, where $M$ is any $t$-qubit matrix.
    This is the set of errors which is relevant if we care about \emph{erasure errors}.
    \item[$\circ$] Given $P \in \{I, X, Y, Z\}^{\otimes n}$,  define the \emph{weight} of $P$ $\mathrm{wt}(P)$ to be the number of non-identity Paulis in $P$'s tensor product.
    Let $\{E_i\}$ contain all the Pauli matrices which have weight at most $t$,
    and set $\calE = \Span\{E_i\}$.
    Then $\calE$ is the set of errors which is relevant if we care about \emph{$t$-qubit errors}.
\end{itemize}
There is a natural generalization of the Pauli matrices to qudits of dimension $d$.
To begin, let us first define the \emph{shift} and \emph{clock} operators as follows.
\begin{equation*}
    X_d \cdot \ket{j} = \ket{j + 1},
    \quad \text{and}
    \quad Z_d \cdot \ket{j} = e^{2 \pi i j/d} \cdot \ket{j}, \quad \text{where $j \in \{0, \ldots, d-1\}$}.
\end{equation*}
Here, the addition $(j+1)$ is performed modulo $d$.
Then the qudit Pauli matrices are given by the set $\{X^a Y^b \mid a, b \in \{0, 1, \ldots, d-1\}\}$.
These form a unitary error set,
and they can be generalized to $n$-qudit matrices by considering tensor products.
\end{example}

However, unitary error sets $\{E_1, \ldots, E_m\}$ can in general model much broader classes of noise than Pauli matrices are able to.
For example, the $E_i$'s need not be low weight or even expressible as a tensor product of single-qudit operators (indeed, we will in general not assume that our Hilbert spaces have any tensor product structure whatsoever).
We believe that identifying unitary error sets as the precise error model that Haar random codes natively error correct is a key conceptual contribution of this work.

\subsubsection{Approximately nondegenerate codes}

Next, we formally define approximately nondegenerate codes.
As far as we know, this notion has not appeared in the literature prior to our work.
To begin, let us recall that an exact quantum code is nondegenerate with respect to a set of errors $\{E_1, \ldots, E_m\}$ if the vectors $E_i \cdot \ket{v_j}$ are orthogonal to each other,
where $\ket{v_1}, \ldots, \ket{v_K}$ is a basis of the codespace.
To define what it means for a code to be approximately nondegenerate, then, let us first define what it means for a set of vectors to be approximately orthogonal.

A set of vectors $\ket{u_1}, \ldots, \ket{u_\ell} \in \C^N$ is \emph{exactly} orthogonal if it satisfies the following condition: for any coefficients $c_1, \ldots, c_{\ell} \in \C$,
the vector
\begin{equation*}
    c_1 \cdot \ket{u_1} + \cdots + c_{\ell}\cdot \ket{u_{\ell}}
    = \Big(\sum_{i=1}^{\ell} \ketbra{u_{i}}{i}\Big) \cdot \Big(\sum_{i=1}^{\ell} c_i \ket{i}\Big)
    \eqqcolon T \cdot \ket{c}
\end{equation*}
has the same length as $\ket{c}$ itself.
This means that the linear transformation $T : \C^{\ell} \rightarrow \C^{N}$ is a length-preserving isometry, and so all of $T$'s singular values are equal to 1.
We now relax these definitions to the approximate case as follows.

\begin{definition}[Approximate isometry]\label{def:approx_isometry}
    Given $\delta \geq 0$, a matrix $T$ of size $N\times \ell$ with $N\geq \ell$ is called a \emph{$\delta$-approximate isometry} if 
    \begin{equation*}
        1-\delta \leq s_{\min}(T) \leq s_{\max}(T)\leq 1+\delta. 
    \end{equation*}
\end{definition}

\begin{definition}[Approximately orthonormal basis]\label{def:approx_ortho_basis}
    Given $\delta \geq 0$, a set of vectors $\ket{u_1}, \ldots, \ket{u_\ell} \in \C^N$ with $N\geq \ell$ forms a \emph{$\delta$-approximately orthonormal basis} if
    \begin{equation*}
        1-\delta \leq \norm{\sum_{i=1}^\ell c_i \ket{u_i}} \leq 1 + \delta
    \end{equation*}
    for any unit vector $\ket{c} = \sum_{i=1}^{\ell} c_i \ket{i} \in \C^{\ell}$.
    This is equivalent to the following statement: if $T = \sum_{i=1}^{\ell} \ketbra{u_i}{i}$, then $T$ is a $\delta$-approximate isometry.
\end{definition}

Note that if $\delta < 1$, then the vectors $\ket{u_1}, \ldots, \ket{u_{\ell}}$ are linearly independent.
Otherwise, there would exist a unit vector $\ket{c} = \sum_{i=1}^{\ell} c_i \ket{i} \in \C^{\ell}$ such that $\sum_{i=1}^\ell c_i \ket{v_i} = 0$, which contradicts the assumption that
\begin{equation*}
    \norm{\sum_{i=1}^\ell c_i \ket{v_i}} \geq 1 - \delta >0. 
\end{equation*}
Our definitions impose a much stronger orthogonality condition than merely bounding the pairwise inner products of vectors. In fact, it is easy to construct a set of vectors with small pairwise inner products which do not form an approximately orthonormal basis.\footnote{For example, for each $i\in [\ell]$, let $\ket{v_i} = \sqrt{1-\delta} \ket{i} + \sqrt{\delta}\ket{0}$.  
Then their pairwise inner products are small: $\braket{v_i}{v_j} = \delta$ if $i\neq j$. However, they will be far from an approximately orthonormal basis when $\ell$ is large, since $\norm{(\ket{v_1} + \cdots + \ket{v_{\ell}})/\sqrt{\ell}} = \sqrt{1 + \delta (\ell - 1)}$.}

With this definition in hand, we can now define what it means for a code to be approximately nonegenerate.

\begin{definition}[Approximate nondegenerate code]
\label{def:approx_nondeg_code}
    Let $V:\C^K\mapsto \C^N$ be an encoding isometry and $\{E_1, \ldots, E_m\}$ be a unitary error set. 
    Given $\delta \geq 0$, $V$ is a \emph{$\delta$-approximate nondegenerate code with respect to $\{E_1, \ldots, E_m\}$} if the vectors $\{E_i \cdot \ket{v_j}\}_{i \in [m], j \in [K]}$ form a $\delta$-approximately orthonormal basis.
    This is equivalent to the statement that the matrix
    \begin{equation*}
        \sum_{i=1}^m \sum_{j= 1}^{K} E_i \ketbra{v_j}{j, i}
        = \sum_{i=1}^m E_i V \otimes \bra{i}
    \end{equation*}
    is a $\delta$-approximate isometry. 
\end{definition}

If $\delta < 1$, then the vectors $E_i \cdot \ket{v_j}$ of a $\delta$-approximate nondegenerate code must be linearly independent.
This implies that $mK \leq N$,
as the overall dimension is $N$.
Thus, the quantum Hamming bound, which holds for exact nondegenerate codes, also applies to $\delta$-approximate nondegenerate codes with $\delta < 1$.

\subsubsection{Main results}

Our first main result states that approximately nondegenerate codes are indeed AQECs.

\begin{restatable}[Approximate nondegenerate codes are AQECs]{theorem}{maindecode}
\label{thm:decode}
    Let $V:\C^K\mapsto \C^N$ be an encoding isometry and $\{E_1, \ldots, E_m\}$ be a unitary error set. 
    Given $0 \leq \delta < 1$, if $V$ is a $\delta$-approximate nondegenerate code with respect to $\{E_1, \ldots, E_m\}$, then $V$ specifies an AQEC which corrects for the set of errors $\calE = \Span\{E_1, \ldots, E_m\}$ with disturbance $\delta$. 
\end{restatable}

Our second main result states that Haar random codes are nondegenerate codes, and therefore AQECs, with high probability,
so long as their parameters approximately satisfy the quantum Hamming bound.

\begin{restatable}[Haar random codes satisfy the quantum Hamming bound]{theorem}{mainhaar}
    \label{thm:haar}
    There is a universal constant $C > 0$ such that the following is true. 
    Let $N,m,K$ be positive integers satisfying
    \begin{equation*}
        \delta \coloneqq 3 \cdot \left( \sqrt{\frac{Km}{N}} + C \cdot \sqrt{\frac{m\cdot (\log N)^3}{N}}\right) < 1. 
    \end{equation*}
    With probability at least $1-2/N^{(\log N)^2}$, 
    a Haar random isometry $\bV:\C^K\mapsto \C^N$ gives a $\delta$-approximate nondegenerate code with respect to a unitary error set $\{E_1, \ldots, E_m\}$.
    As a result, $\bV$ gives an AQEC which corrects for the set of errors $\calE = \Span\{E_1, \ldots, E_m\}$ with disturbance $\delta$.
\end{restatable}

To understand our bound,
suppose that $K \geq (\log N)^3$,
meaning that the codespace is assumed to have a small lower bound on its dimension.
Then the Haar random code can be corrected with disturbance $O(\sqrt{Km/N})$, 
which is small so long as $Km \ll N$.
This is the sense in which Haar random codes attain the quantum Hamming bound, approximately.
\Cref{thm:haar} entails showing that the set of vectors $\{E_i \cdot \ket{\bv_j}\}_{i \in [m], j \in [K]}$, where $\ket{\bv_j} = \bV \cdot \ket{j}$, are approximately orthonormal so long as $mK \ll N$.
This means that a small ``seed'' of Haar random vectors $\{\ket{\bv_j}\}_{j \in [K]}$ can be expanded into a large set of approximately orthonormal vectors simply by shifting them by a fixed set of unitary matrices $\{E_i\}_{i \in [m]}$, a statement that we found surprising.

We note that even if $K$ is not at least $(\log N)^3$, our bound still shows that Haar random codes are good AQECs when $mK \ll N$ and $m \ll N / (\log N)^3$, which is a relatively mild condition on the number of errors $m$.
It is not clear to us if this second condition on $m$ is necessary, or if it is merely a limitation of our techniques.

\subsubsection{Consequences of our results}

By far the most widely studied setting in quantum error correction is when the code consists of multiple qudits and errors are local, affecting only a bounded number of qudits at a time.
Here, we will illustrate the power of our results by instantiating them in this setting and comparing them to known bounds on QECs.
To begin, let us recall the definition of an exact $[[n, k, d]]_q$ QEC: 
this is a code $C$ which encodes $k$ qudits of dimension $q$ into $n$ such qudits,
meaning the encoding isometry is of the form $V : (\C^q)^{\otimes k} \rightarrow (\C^q)^{\otimes n}$, and which has distance $d$.
In terms of correcting errors, having distance $d$ can be interpreted in one of two ways.
\begin{enumerate}
    \item[$\circ$] First, if $d = 2t+1$, then it means that the code can be corrected from any $t$-local error, which is any error in $\Span\{P \mid \text{$P$ is a weight-$t$ Pauli}\}$.
    \item[$\circ$] Second, it means that the code can be corrected from any $(d-1)$-qudit erasure error. If $\ket{\psi} \in C$ is a code state, such an error involves choosing a subset $S \subseteq[n]$ of the qudits of size at most $(d-1)$ and discarding them, so that the decoding algorithm is provided $\tr_S(\ketbra{\psi})$ as well as the subset $S$; it must then use these to recover $\ket{\psi}$. It turns out to be equivalent that the decoding algorithm can recover from any error in $\Span\{P \mid \mathrm{supp}(P) \subseteq S\}$, so long as $S$ is provided.
\end{enumerate}
Curiously, in the setting of \emph{approximate} quantum error correction, it is unclear whether there is still an operationally meaningful notation of distance. As Cr\'{e}peau, Gottesman, and Smith put it,  ``This suggests there is no sensible notion of distance for an approximate quantum error-correcting code.''~\cite{CGS05}.\footnote{Bergamaschi, Golowich, and Gunn have proposed defining the distance of an AQEC to be one plus the number of errors it can correct approximately~\cite{BGG24}. Although this does indeed generalize one way to define the distance for an exact QEC, it is unclear to us whether this can be naturally interpreted in terms of a distance between codewords.}

Now we recall two well-known bounds on the parameters of an exact QEC.
First, the quantum Singleton bound states that any $[[n, k, d = 2t+1]]_q$ exact QEC must satisfy $d-1 \leq (n-k)/2$ or, equivalently, $t \leq (n-k)/4$.
This means that an exact QEC can only correct erasure errors on at most $(n-k)/2$ of the qudits,
and it can only correct $t$-local errors for $t \leq (n-k)/4$.
Writing $R = k/n$ for the \emph{rate} of the code,
we can restate these bounds as follows:
an exact QEC can handle erasure errors on $n/2 \cdot (1-R)$ qudits and $t$-local errors for $t \leq n/4 \cdot (1-R)$.

Next, as we have seen, the quantum Hamming bound states that a nondegenerate QEC of dimension $K$ correcting $m$ errors must satisfy $mK \leq N$, where $N$ is the dimension of the overall space.
An $[[n, k, 2t+1]]_q$ exact nondegenerate QEC has $N = q^n$ and $K = q^k$. Furthermore, since it corrects all weight-$t$ Paulis, we have
\begin{equation*}
    m = \sum_{i=0}^t \binom{n}{i} \cdot (q^2 - 1)^i.
\end{equation*}
Hence, the quantum Hamming bound states that $m \leq q^{n-k}$.
Taking the logarithm of both sides, we have $\log_q(m) \leq n - k = n \cdot (1-R)$.

A direct consequence of our \Cref{thm:haar} is that Haar random codes will, with exponentially small disturbance and exponentially small chance of failure, essentially match the quantum Singleton bound for erasure errors and the quantum Hamming bound for general weight-$t$ errors.
In addition, our \Cref{thm:haar}  also implies that Haar random codes are essentially able to correct \emph{twice} as many general qudit errors as exact QECs are able to, due to the quantum Singleton bound.
We show these consequences below.

\begin{corollary}[Consequences of our main result]\label{cor:main_bounds}
    Let $q$ be the qudit local dimension,
    and let $k \leq n$ be positive integers.
    Set $N = q^n$, $K = q^k$, and $R = k/n$. Suppose that $K\geq (\log_2 N)^3$, i.e.\ $(q^{R/3})^n \geq n\cdot \log_2(q)$.
    Then the following statements are true.

    \begin{enumerate}
        \item (Quantum Singleton bound for erasure errors) \label{item:erasure}
        With probability at least $1-2^n\cdot 2/N^{(\log N)^2}$, a Haar random code which encodes $k$ qudits into $n$ qudits can correct for erasure errors on at most $t$ qudits with disturbance $O(q^{-\gamma\cdot n/2})$ if
        \begin{equation*}
            t \leq \frac{n}{2}\cdot \left(1-R - \gamma\right), \qquad \text{for any }\gamma > 0. 
        \end{equation*}
        
        \item (Twice the quantum Singleton bound for general qudit errors)\label{item:double-single}
        With probability at least $1-2/N^{(\log N)^2}$, a Haar random code which encodes $k$ qudits into $n$ qudits can correct $t$-local errors with disturbance $O(2^{-n/2})$ if
        \begin{equation*}
            t \leq \frac{n}{2}\cdot \left(1 - R - \frac{2}{\log_2(q)} \right).  
        \end{equation*}

        \item (Quantum Hamming bound for general qudit errors) 
        With probability at least $1-2/N^{(\log N)^2}$, a Haar random code which encodes $k$ qudits into $n$ qudits can correct $t$-local errors with disturbance $O(q^{-\gamma \cdot n/2})$ if
        \begin{equation*}
            \log_q(m) \leq n\cdot (1-R-\gamma), \qquad \text{for any }\gamma>0,
        \end{equation*}
        where
        \begin{equation*}
            m = \sum_{i=0}^t \binom{n}{i}\cdot (q^2 - 1)^i. 
        \end{equation*}
    \end{enumerate}
\end{corollary}

\begin{proof}
    Since we are assuming that $K\geq (\log_2 N)^3$, if we apply \Cref{thm:haar} with a unitary error set of size~$m$, the disturbance will be bounded by $O(\sqrt{Km/N})$. 
    Now we consider the three cases.
    \begin{enumerate}
        \item In an erasure error,
        one is provided the subset $S$ of size $|S| \leq t$ on which the erasure has occurred.
        Writing
        \begin{equation*}
            \calP_S = \{\text{qudit Paulis $P$} \mid \mathrm{supp}(P) \subseteq S\} \quad\text{and} \quad \calE_S = \Span\{\calP_S\},
        \end{equation*} 
        we can model an erasure error as a channel in $\channels{\calE_S}$, and therefore it suffices to correct for the set of errors $\channels{\calE_S}$.
        There are $q^2$ single-qudit Paulis $X^a Z^b$, and so $|\calP_S| = q^{2t}$.
        As a result, \Cref{thm:haar} states that probability at least $1 - 2/N^{(\log N)^2}$, a Haar random code will correct for the errors in $\calE_S$ with disturbance at most
        \begin{equation*}
            O(q^{\frac12 \cdot (k-n+2t)}) = O(2^{-\gamma n/2}).
        \end{equation*}
        Since $S$ may be chosen arbitrarily, we now union bound over all subsets $S$, of which there are trivially at most $2^n$.
        Hence, with probability $1 - 2^n \cdot 2/N^{(\log N)^2}$, we can correct for any $t$-qudit erasure error with the above disturbance.
        \item Consider the set of qudit Paulis with weight at most $t$. We can bound its size $m$ by
        \begin{equation*}
            m \leq \binom{n}{t}\cdot q^{2t} \leq 2^n \cdot q^{2t} = q^{2t + n/\log_2(q)}.
        \end{equation*}
        Applying \Cref{thm:haar}, with probability at least $1 - 2/N^{(\log N)^2}$, a Haar random code will correct for all $t$-local errors with disturbance
        \begin{equation*}
        O\left(q^{\frac12\left(k - n + 2t + \frac{n}{\log_2(q)}\right)}\right) = O\left(q^{\frac{1}{2} \cdot \frac{n}{\log_2(q)} (-2 + 1) }\right) = O(2^{-n/2})
        \end{equation*}
        
        \item By assumption, $m$ is the number of qudit Paulis with weight at most $t$. Applying \Cref{thm:haar}, with probability at least $1 - 2/N^{(\log N)^2}$, a Haar random code will correct for all $t$-local errors with disturbance
        \begin{equation*}
            O(q^{\frac12 \cdot (k-n+\log_q(m)}) = O(q^{-\gamma n /2}).
        \end{equation*}
    \end{enumerate}
    This completes the proof.
\end{proof}

\subsection{Related work}

\paragraph{Approximate quantum error correction.}

Bergamaschi, Golowich, and Gunn~\cite{BGG24} gave two constructions of AQECs which essentially saturate twice the quantum Singleton bound.
Their constructions give AQECs which encode $k$ qudits into $n$ qudits and can correct $t$-local errors if
\begin{equation*}
    t \leq \frac{n}{2} \cdot (1 - R - \gamma),
\end{equation*}
for a parameter $\gamma > 0$.
Their first construction uses qudits of local dimension $q = 2^{O(1/\gamma^5)}$ and has disturbance $2^{-\Omega(n)}$,
and their second construction uses qudits of slightly smaller local dimension $q = 2^{O(1/(\gamma^4\log(1/\gamma))}$ at the expense of a slightly worse disturbance of $2^{-O(\gamma n)}$.
Our \Cref{cor:main_bounds} shows that these guarantees are not actually specific to their constructions but are actually properties of \emph{most} AQECs.
Indeed, our results improve on their bounds slightly: \Cref{item:double-single} of \Cref{cor:main_bounds} shows that a Haar random code will, with high probability, correct the same number of errors as their codes do, with disturbance $2^{-\Omega(n)}$ and an improved local dimension of $2^{O(1/\gamma)}$.
However, one advantage of their codes is that their encoding and decoding operations are computationally efficient.

A follow-up work by Bergamaschi~\cite{Ber24} gave a construction of an AQEC which encodes $k$ qudits into $n$ qudits and can correct erasure errors on up to
\begin{equation*}
            t \leq \frac{n}{2}\cdot \left(1-R - \gamma\right)
        \end{equation*}
qudits with disturbance $2^{-\Omega(\gamma \cdot n)}$, for any $\gamma \geq 1/\log(n)$.
As above, our results show that these guarantees are actually properties of generic, Haar random codes, and that Haar random codes actually achieve slightly improved guarantees: by \Cref{item:erasure} of \Cref{cor:main_bounds}, a Haar random qudit code can correct the same number of erasures with an improved disturbance of $q^{-\Omega(\gamma \cdot n)}$ and $\gamma$ allowed to be any number $> 0$. However, as above, the Bergamaschi code has the advantage of having efficient encoding and decoding operations.

Finally, we note the work of Mamindlapally and Winter~\cite{MW23}, which has shown an analogue of the quantum Singleton bound for AQECs.
In particular, their Theorem 5, as stated in \cite[Theorem 1.6]{BGG24}, implies that any AQEC of rate $R = k/n$ correcting $t$-local errors with disturbance $\delta$ must satisfy
\begin{equation*}
    t \leq \frac{n}{2} \cdot\Big(1 - R + O\Big(\delta + \frac{1}{n} \cdot \delta \log_q(1/\delta)\Big)\Big).
\end{equation*}
We note that in our setting, when the disturbance $\delta$ is exponentially small in $n$, this bound actually implies the seemingly stronger bound of $t \leq n/2 \cdot (1- R)$.
This is because $t$ must always be an integer, and so the exponentially-small factor of $O(\delta \cdot n) + O(\delta \log_q(1/\delta))$ may be dropped, at least when $n$ is large enough to make this term strictly smaller than 1.
Hence, \Cref{item:double-single} of our \Cref{cor:main_bounds} is optimal, up to the ``$2/\log_2(q)$ term'' that we can suppress by taking $q$ large.

\paragraph{Exact quantum error correction.}
An exact, nondegenerate quantum error correction code is called \emph{perfect} if it exactly attains the quantum Hamming bound. In other words, an $[[n, k, d = 2t+1]]_q$ QEC is perfect if $m K = N$, for $K = q^k$, $N = q^n$, and
\begin{equation*}
    m = \sum_{i=0}^t \binom{n}{i} \cdot (q^2 - 1)^i.
\end{equation*}
Perfect codes, such as the quantum Hamming codes and the quantum twisted codes, are known to exist, but only for certain specific settings of parameters.
Indeed, the entire set of possible parameters a perfect quantum code can attain has been classified~\cite{LX09};
this classification states that a perfect $[[n, k, d]]_q$ QEC exists only if
\begin{equation*}
    n = \frac{q^{2 \ell} - 1}{q^2 - 1}, \quad
    k = n - 2 \ell, \quad\text{and }
    d = 3,
\end{equation*}
for some integer $\ell \geq 2$.  Hence, perfect exact QECs are only capable of correcting $1$-local qudit errors.
Our work highlights the power of approximate error correction, as it allows us to approximately saturate the quantum Hamming bound for a much wider range of parameters.

\paragraph{Unitary error sets.}
To our knowledge, unitary error sets have not appeared in the literature prior to our work.
In the special case when $E_1, \ldots, E_m$ also form a basis for the space of linear operators acting on $\calH$, however, then they have previously appeared in the literature under the name of \emph{unitary error bases}.
As the name suggests, unitary error bases have been studied in quantum error correction (for example, to understand the ``right'' way to model single-qudit errors in order to generalize the single-qubit Pauli errors~\cite{Kni96}),
but they have also been studied in areas such as quantum teleportation and superdense coding~\cite{Wer01,NY23}.
A central goal of this line of research has been to construct and classify the different unitary error bases~\cite{VW00,MV16}, and it is known, for example, that there exist multiple unitary error bases which are not equivalent to each other under any unitary change-of-basis~\cite{Wer01}. (And thus, there exist unitary error bases which are not just ``rotated Pauli matrices''.)
As far as we can tell, unitary error sets strictly generalize the notion of unitary error bases: for example, it is not obvious whether every unitary error set $E_1, \ldots, E_m$ can be extended to a unitary basis. Thus, there may even be unitary error sets which do not appear as subsets of any unitary error basis.

\subsection{Notation}

Throughout this paper, we will write $\{E_1, \ldots, E_m\}$ for an arbitrary unitary error set on $\C^N$ and set $\calE = \Span\{E_1, \ldots, E_m\}$. 
We will always assume that $N, K, m$ are positive integers satisfying $Km \leq N$.

We use \textbf{boldface} to denote random variables, 
$\mathrm{i} = \sqrt{-1}$ to denote the imaginary unit, $I_d$ to denote the $d\times d$ identity matrix, and
$[d] = \{1, \dots, d\}$. We use $s_{\max}(\cdot)$ and $s_{\min}(\cdot)$ to denote the largest and smallest singular value of a matrix respectively. 
We will write $\norm{\ket{\psi}}$ for the $2$-norm of a vector $\ket{\psi}$, $\norm{A}$ for the operator norm of a matrix $A$, $\norm{\cdot}_1$ for the trace norm, and $\norm{\cdot}_\diamond$ for the diamond norm.

\section{Approximate nondegenerate codes are AQECs}\label{sec:approx-nondegenerate-codes-are-aqecs}

The goal of this section is to prove that any approximate nondegenerate code can be approximately decoded. 

Let $V: \mathbb{C}^K \mapsto \mathbb{C}^N$ be the encoding isometry for the code. $V$ can be written as $\sum_{j \in [K]} \ketbra{v_j}{j}$ where the vectors $\ket{v_1},\dots,\ket{v_K}$ form an orthonormal basis for the codespace. The code is a $\delta$-approximate nondegenerate code with respect to the unitary error set $\{E_1,\dots,E_m\}$ if and only if $\{E_i \cdot \ket{v_j}\}_{i \in [m],j \in [K]}$ is a $\delta$-approximately orthonormal basis (\Cref{def:approx_ortho_basis}).

To show that this code can be approximately decoded, we will define an explicit decoding channel $\Dec$, which depends on $V$ and the unitary error set $\{E_i\}_{i \in [m]}$. The main theorem of this section is that $\Dec$ approximately recovers any quantum state that has been encoded by $V$ and corrupted by noise channels supported on the error set $\{E_i\}_{i \in [m]}$.

\subsection{The decoding channel and the error-diagnosing isometry}

In this subsection, we will define the decoding channel $\Dec$.

\paragraph{Intuition for the decoding channel: the perfectly nondegenerate case} To give intuition for how we define $\Dec$, suppose that the code is perfectly nondegenerate, i.e.\ the vectors $\{E_i \cdot \ket{v_j}\}_{i \in [m],j \in [K]}$ are perfectly orthogonal. Suppose we encode a state $\ket{\phi} = \sum_{j \in [K]} \alpha_j \ket{j}$ using the isometry $V$, obtaining
\begin{equation*}
    V \ket{\phi} = V \cdot \sum_{j \in [K]} \alpha_j \ket{j} = \sum_{j \in [K]} \alpha_j \ket{v_j}.
\end{equation*}
Suppose an error $F = \sum_{i \in [m]} c_i E_i$ occurs on $V \ket{\phi}$, resulting in the state 
\begin{equation*}
    F \cdot V \ket{\phi} = \Big( \sum_{i \in [m]} c_i E_i \Big) \Big( \sum_{j \in [K]} \alpha_j \ket{v_j} \Big) = \sum_{i \in [m], j \in [K]} c_i \alpha_j E_i \ket{v_j}.
\end{equation*}
To decode this state, we first define the decoding isometry
    \begin{equation*}
        D = \sum_{i \in [m], j \in [K]} \ketbra{j,i}{v_j} E_i^\dagger,
    \end{equation*}
    which maps $E_i \ket{v_j} \mapsto \ket{j} \ket{i}$ for all $i, j$. Technically, $D: \mathbb{C}^N \mapsto \mathbb{C}^K \otimes \mathbb{C}^m$ is a \emph{partial} isometry: it is an isometry mapping the $Km$-dimensional subspace $\Im(D^\dagger) = \Span\{E_i \ket{v_j}: i \in [m],j \in [K]\}$ of $\mathbb{C}^N$ to the space $\mathbb{C}^{K} \otimes \mathbb{C}^{m}$, but acts as the $0$ operator on the subspace of $\mathbb{C}^N$ orthogonal to $\Im(D^\dagger)$. 
    
Using the decoding isometry $D$, we can decode the state $F \cdot V \ket{\phi}$ as follows:
\begin{enumerate}
    \item
    Apply $D$, which yields
    \begin{equation*}
        D \cdot \Big( F \cdot V \ket{\phi}\Big) = D \cdot \Big(\sum_{i \in [m], j \in [K]} c_i \alpha_j E_i \ket{v_j}\Big) = \sum_{j \in [K]} \alpha_j \ket{j} \otimes \sum_{i \in [m]} c_i \ket{i}
    \end{equation*}
    \item Trace out the $\ket{i}$ register. This exactly recovers the original state $\sum_{j \in [K]} \alpha_j \ket{j} = \ket{\phi}$.
\end{enumerate}

\paragraph{The approximate case.} In our setting, the states $\{E_i \ket{v_j}\}$ are not guaranteed to be orthogonal, and thus
\begin{equation*}
    \widehat{D} \coloneqq \sum_{i \in [m], j \in [K]} \ketbra{j,i}{v_j} E_i^\dagger = \sum_{i\in [m]} V^\dagger E_i^\dagger \otimes \ket{i},
\end{equation*}
is not necessarily a partial isometry, i.e.\ $\widehat{D}$ does not correspond to a physically allowable operation. 

However, if the vectors $\{E_i \ket{v_j}\}_{i \in [m], j \in [K]}$ are a $\delta$-approximately orthonormal basis, then the singular values of $\widehat{D}$ lie between $1-\delta$ and $1+\delta$. This motivates defining our decoding isometry $D$ by rounding the singular values of $\widehat{D}$ as follows:
\begin{itemize}
    \item Consider the singular value decomposition of $\widehat{D} = U_1 \cdot \widehat{\Sigma} \cdot U_2$. By assumption, $\widehat{\Sigma}$ is a rectangular diagonal matrix whose diagonal entries are between $1-\delta$ and $1+\delta$.
    \item Let $\Sigma$ be the matrix obtained by replacing the non-zero entries of $\widehat{\Sigma}$ with $1$. Define $D\coloneqq U_1 \Sigma U_2$.
\end{itemize}
Note that $D$ is close to $\widehat{D}$ in operator norm, since
\begin{equation}
    \norm*{\widehat{D} - D} = \norm*{U_1 (\widehat{\Sigma} - \Sigma) U_2} = \norm*{\widehat{\Sigma} - \Sigma } \leq \delta .\label{eq:op-norm-D-Dhat} 
\end{equation}

\begin{proposition}
\label{prop:D-partial-iso}
$D$ is a partial isometry where $\Im(D) = \mathbb{C}^K \otimes \mathbb{C}^m$.
\end{proposition}

\begin{proof}
    Recall that an operator $D$ is a partial isometry if and only if all of its singular values are $0$ or $1$. This is guaranteed by the way we constructed $D$. It remains to show that $\Im(D) = \mathbb{C}^K \otimes \mathbb{C}^m$.
    
    By assumption, $\widehat{D}$ is a $\delta$-approximate isometry, and so all $K\cdot m$ of its singular values are non-zero (since they are bounded between $1-\delta$ and $1+\delta$ for $\delta < 1$). Thus, $\Im(\widehat{D})$ corresponds to the column span of the unitary $U_1$, which equals $\mathbb{C}^K \otimes \mathbb{C}^m$. Finally, the construction of $\widehat{D}$ guarantees that $\Im(D) = \Im(\widehat{D})$. 
\end{proof}

Given a noisy codestate, which is guaranteed to be in $\Im(D^\dagger) = \Span\{E_i \ket{v_j}: i \in [m], j \in [K]\}$, we can decode using the same strategy that we used in the perfect orthogonality case: apply the decoding isometry $D$ and trace out the $\ket{i}$ register. For completeness, we define our decoding channel $\Dec$ to act not just on states in $\Im(D^\dagger)$, but on all of $\mathbb{C}^N$. Our full decoding procedure $\Dec$ is as follows:
\begin{definition}
\label{def:recovery}
The decoding channel $\Dec$ works as follows:
\begin{enumerate}
    \item Apply the binary projective measurement $\{D^\dagger D, I_N - D^\dagger D\}$ to test if the state is in $\Im(D^\dagger)$. If not, output an arbitrary state in $\mathbb{C}^K$ (say, a maximally mixed state) and skip the remaining steps.
    \item Apply the decoding isometry $D$.
    \item Trace out the $\ket{i}$ register. 
\end{enumerate}
\end{definition}

\subsection{Recovering from general noise channels}

We now prove that our decoding channel $\Dec$ works for an arbitrary noise channel $\calN$ with Kraus operators $\{K_r\}_r$ where $\sum_r K_r^\dagger K_r = I$ and $K_r = \sum_{i\in [m]} c_{r,i} E_i$ for some coefficient $c_{r,i}$. It will be convenient to state our technical lemma in terms of the Stinespring dilation of $\calN$. That is, let $E_{\calN}$ be the isometry that maps $\ket{\psi}$ to $\sum_r K_r \ket{\psi} \ket{r}$. The channel $\calN$ can be implemented by applying the isometry $E_{\calN}$ and tracing out $\ket{r}$. The fact that $E_{\calN}$ is an isometry readily follows from the condition that $\sum_r K_r^\dagger K_r = I$. 

\begin{proposition}
    When $\{E_i\}_{i}$ is a unitary error set, the coefficients $c_{r,i}$ satisfy $\sum_{r,i} \abs{c_{r,i}}^2 = 1$.
\end{proposition}

\begin{proof}
    Taking the trace of both sides of the equation $\sum_r K_r^\dagger K_r = I$, we have $\Tr(\sum_r K_r^\dagger K_r) = \Tr(I) = N$. Then, expanding $\Tr(\sum_r K_r^\dagger K_r)$ in terms of the coefficients $c_{r,i}$ yields
\begin{equation*}
    \Tr\Big(\sum_r K_r^\dagger K_r\Big) = \sum_r \Tr\Big(\Big(\sum_i c_{r,i}E_i\Big)^\dagger \Big(\sum_{i'} c_{r,i'}E_{i'}\Big)\Big) = \sum_r \sum_{i,i'} \overline{c_{r,i}} c_{r,i'} \Tr(E_i^\dagger E_{i'}) = \sum_{r,i} \abs{c_{r,i}}^2 \cdot N,
\end{equation*}
    which completes the proof.
\end{proof}

\begin{lemma}
\label{lemma:noise-isometry}
    Let $V$ be the encoding isometry of a $\delta$-approximate nondegenerate code with respect to a unitary error set $\{E_i\}_i$.
    Fix a noise isometry $E_{\calN}: \ket{\psi} \mapsto \sum_r K_r \ket{\psi}\ket{r}$, where each $K_r = \sum_i c_{r,i} E_i$ for some coefficients $c_{r,i}$. Fix an input state $\ket{\phi}_{\sA,\sB}$, where $\sA$ is a $K$-dimensional Hilbert space corresponding to the message register of the code and $\sB$ is some ancillary system. Let $\ket{c} = \sum_{r,i} c_{r,i} \ket{i}\ket{r}$. Then
    \begin{equation}
        \norm{D \cdot E_{\calN} \cdot V \cdot \ket{\phi} - \ket{\phi} \ket{c} } \leq \delta,\label{eq:nondegen-to-aqec}
    \end{equation}
    where $D$ acts as identity on the $\ket{r}$ register and $D, E_{\calN}, V$ all act as identity on system $\sB$. 
\end{lemma}

\begin{proof}
    Using the definitions of $E_{\calN}$ and $\widehat{D}$, 
    \begin{equation*}
        D \cdot E_{\calN} \cdot V \cdot \ket{\phi} = D \cdot \sum_{r,i} c_{r,i} E_i \cdot V \cdot \ket{\phi} \ket{r} = D \cdot \widehat{D}^\dagger \ket{\phi} \ket{c}.
    \end{equation*}
    Then the left-hand side of~\Cref{eq:nondegen-to-aqec} can be written as
    \begin{equation*}
        \norm{D \cdot \widehat{D}^\dagger \ket{\phi} \ket{c} - \ket{\phi}\ket{c}} \leq \norm{D \cdot \widehat{D}^\dagger - I} \leq \delta ,
    \end{equation*}
where the last inequality uses the fact that $D \cdot D^\dagger = \Im(D) = I_{Km}$ (by~\Cref{prop:D-partial-iso}) and that $\norm*{\widehat{D}^\dagger - D^\dagger} \leq \delta$ (previously established in~\Cref{eq:op-norm-D-Dhat}).
\end{proof}

We are now ready to prove our first main result, \Cref{thm:decode}, restated below for convenience. 

\maindecode*
\begin{proof}
    Fix any state $\ket{\phi}_{\sA,\sB}$ where $\sA$ is a $K$-dimensional Hilbert space corresponding to the message register of the code, and let $E_{\calN}$ be the Stinespring dilation of $\calN$. Since $\norm{\ketbra{u} - \ketbra{v}}_1 \leq 2 \cdot \norm{\ket{u} - \ket{v}}$, we have 
    \begin{align*}
        \norm{D \cdot E_{\calN} \cdot V \cdot \ketbra{\phi} \cdot V^\dagger \cdot E_{\calN}^\dagger \cdot D^\dagger - \ketbra{\phi} \otimes \ketbra{c}}_1 &\leq 2\cdot \norm{D \cdot E_{\calN} \cdot V \cdot \ket{\phi} - \ket{\phi} \ket{c} } \\
        &\leq 2\delta. \tag{\Cref{lemma:noise-isometry}}
    \end{align*}
    Let $\mathsf{C}$ denote the register corresponding to $\ket{c}$. Note that the channel $\Dec \circ \calN \circ \Enc$ is equivalent to applying $D \cdot E_{\calN} \cdot V$ and tracing out $\sC$. Since tracing out the $\sC$ system cannot increase the $1$-norm, we have
    \begin{equation*}
        \max_{\ket{\phi}_{\sA,\sB}} \norm{(\Dec\circ\calN \circ \Enc \otimes I_{\sB})(\ketbra{\phi}) - \ketbra{\phi}_{\sA,\sB}}_1 \leq 2 \delta,
    \end{equation*}
    which implies the diamond distance bound $\norm{\Dec \circ \calN \circ \Enc - \calI}_\diamond \leq 2 \delta$, so the disturbance is at most $\delta$.
    \end{proof}

\section{Haar random codes are approximate nondegenerate codes}\label{sec:proof_shifted_basis}

The goal of this section is to prove~\Cref{thm:haar} that Haar random codes are approximate nondegenerate codes. Concretely, we will prove that if $\bV: \mathbb{C}^{K} \mapsto \mathbb{C}^{N}$ is a Haar random isometry and $\{E_1,\dots,E_m\}$ is a unitary error set, then the matrix
\begin{equation*}
    \bY = \sum_{i=1}^{m} E_i \bV \otimes \bra{i}
\end{equation*}
is an approximate isometry with high probability.

Our proof will leverage the fact that when $K \ll N$, a Haar random isometry $\bV$ is well-approximated by an (appropriately scaled) matrix of i.i.d.\ complex Gaussians. In particular, the first step of our proof will be to replace $\bV$ with $\bG$, an $N \times K$ dimensional matrix where each entry is an i.i.d.\ complex Gaussian with mean $0$ and variance $1/N$, yielding the matrix
\begin{equation*}
    \bX = \sum_{i=1}^{m} E_i  \bG \otimes \bra{i}.
\end{equation*}
This is an $N \times m N$ matrix whose entries are correlated Gaussians, and our goal is to show that its singular values are close to~1 with high probability.
There is a large body of literature in random matrix theory studying questions of precisely this form,
and the high level message of this literature is that proving bounds on the singular values is easy when the matrix entries have a simple covariance structure (e.g.\ when the entries are independent and identically distributed Gaussians), but it becomes more difficult the more complicated the covariances become.
In our case, we were able to prove that $\bX$ is an approximate isometry when $mK \ll N /\polylog(N)$ using standard random matrix theory tools such as decoupling, the matrix Chernoff bound, and the matrix Khintchine inequality. (See~\cite{Van17} for an overview of these standard tools.)
However, achieving the optimal bound of $mK \ll N$ required using a recent and powerful matrix concentration result of Bandeira, Boedihardjo, and van Handel~\cite{BBV23},
which can provide strong spectral norm bounds on Gaussian matrices with arbitrary covariance structures.

To finish the proof, we use the fact that a Haar random isometry $\bV$ can be sampled by first sampling a random Gaussian matrix $\bG$, and then applying an ``isometrize'' operation which rounds all of the singular values of $\bG$ to $1$ yields a Haar-distributed isometry $\bV$. We then prove~\Cref{lemma:isometrize}, which says that whenever $X = \sum_i E_i G \otimes \bra{i}$ is an approximate isometry, then the matrix $Y = \sum_i E_i V \otimes \bra{i}$, where $V$ is obtained by ``isometrizing'' $G$, is also an approximate isometry.

\subsection{Matrix concentration inequalities}

We will first introduce the Gaussian concentration inequalities developed in \cite{BBV23} which are crucial for our proofs.  
Let $\bX$ be a $d\times m$ random matrix with correlated Gaussian entries, where $d\geq m$, and suppose we are interested in proving concentration bounds on $s_{\max}(\bX)$ and $s_{\min}(\bX)$.
As correlated Gaussians can always be written as linear combinations of uncorrelated Gaussians,
we can model $\bX$ as a sum of independent standard Gaussian variables $\bg_i$ with fixed matrix coefficients $A_i$, i.e.\ $\bX = \sum_i \bg_i \cdot A_i$.
Our goal is to bound
\begin{equation*}
    s_{\max}(\bX) = \Vert \bX^{\dagger} \bX \Vert^{1/2}
\end{equation*}
with high probability,
and a common heuristic in random matrix theory is that this should typically be close to the same expression with an expectation inside the norm, i.e.
\begin{equation*}
    \Vert \E[\bX^{\dagger} \bX] \Vert^{1/2}
    = \Big\Vert \sum_i A_i^{\dagger} A_i \Big\Vert.
\end{equation*}
The main result of~\cite{BBV23} is that $s_{\max}(\bX)$ is indeed close to this value with high probability, up to some lower order terms which will be small in our case.
In particular, we will make use of the following theorem, which we will show can be extracted from the results of \cite{BBV23}.  

\begin{theorem}[Gaussian concentration inequalities]\label{thm:BBvH23_matrix_concentration}
    Let $A_1, \ldots, A_n$ be arbitrary $d\times m$ matrices with complex entries, where $d\geq m$. 
    Let $\bg_1,\ldots, \bg_n$ be i.i.d.\ real Gaussian variables with zero mean and unit variance.
    Define
    \begin{equation*}
        \bX \coloneqq \sum_{i=1}^n \bg_i \cdot A_i. 
    \end{equation*}
    Let us further define
    \begin{align*}
        \sigma(\bX)^2 &\coloneqq \max \left\{\E[\bX^\dagger \bX], \E[\bX \bX^\dagger]\right\} = \max \left\{ \norm{\sum_{i=1}^n A_i^\dagger A_i} , \quad \norm{\sum_{i=1}^n A_i A_i^\dagger} \right\}, \\
        v(\bX)^2 &\coloneqq \norm{\mathrm{Cov}(\bX)},
    \end{align*}
    where $\mathrm{Cov}(\bX)$ is the $dm \times dm$ covariance matrix of the $dm$ entries of $\bX$, i.e.\ $\mathrm{Cov}(\bX)_{ij,k\ell} = \E[\bX_{ij} \bX_{k\ell}^\dagger]$ for all $i,k\in [d]$ and $j,\ell\in [m]$. 
    Then there exists a universal constant $C>0$ such that for all $t\geq 0$, with probability at least $1-2\exp(-t^2)$ over the randomness in $\bX$, 
    \begin{equation*}
       \sqrt{s_{\min} \left(\sum_{i=1}^n A_i^\dagger A_i \right)} - \delta \leq s_{\min}(\bX) \leq s_{\max}(\bX) \leq \sqrt{s_{\max} \left(\sum_{i=1}^n A_i^\dagger A_i \right)} + \delta, 
    \end{equation*}
    where
    \begin{equation*}
        \delta = \norm{\sum_{i=1}^n A_i A_i^\dagger}^{1/2} + C \cdot \left(\sigma(\bX)^{1/2}\cdot v(\bX)^{1/2}\cdot (\log d)^{3/4} + v(\bX)\cdot t \right). 
    \end{equation*}
\end{theorem}

\begin{proof}
    The lower bound for $s_{\min}(\bX)$ follows from \cite[Theorem 3.14 and Lemma 3.15]{BBV23}. 

    The upper bound for $s_{\max}(\bX)$ when $d=m$ follows from \cite[Corollary 2.2 and Lemma 2.5 (Pisier)]{BBV23}. 
    The case of $A_i$'s being non-square matrices, i.e.\ $d>m$, can be reduced to the square case by padding each $A_i$ with $d-m$ columns of zeros on the right. 
    More concretely, let us denote these zeros-padded square matrices as $B_i$'s and let $\bX' \coloneqq \sum_{i=1}^n \bg_i \cdot B_i$. 
    Then we have the following observations: 
    \begin{enumerate}[label=(\arabic*)]
        \item $s_{\max}(\bX') = s_{\max}(\bX)$. 
        \item $\sum_{i=1}^n B_i B^\dagger_i = \sum_{i=1}^n A_i A_i^\dagger$. 
        \item $\sum_{i=1}^n B_i^\dagger B_i$ is a $d\times d$ matrix where the top left $m\times m$ block equals $\sum_{i=1}^n A_i^\dagger A_i$ and all other entries are zeros. 
        \item $\text{Cov}(\bX')$ is a $d^2 \times d^2$ matrix where the top left $dm\times dm$ block equals $\text{Cov}(\bX)$ and all the other entries are zeros, i.e.\ for all $i,j,k,\ell \in [d]$, 
        \begin{equation*}
            \text{Cov}(\bX')_{ij, k\ell} = \begin{cases}
                \text{Cov}(\bX)_{ij, k\ell} & \text{if }j\leq m, \ell \leq m, \\
                0 & \text{else}. 
            \end{cases}
        \end{equation*}
    \end{enumerate}
    We note that $(2)$ and $(3)$ imply that $\sigma(\bX') = \sigma(\bX)$, and $(4)$ implies that $v(\bX') = v(\bX)$. 
\end{proof}

\subsection{Approximate isometries from Gaussian random matrices}

\begin{definition}[Complex Gaussian random variable]
    We say that $\bg$ is a \emph{complex Gaussian random variable with mean $\mu$ and variance $\sigma^2$} if
    \begin{equation*}
        \bg = \frac{1}{\sqrt{2}}\cdot \left(\bg^{\R} + \bg^{\C}\cdot \mathrm{i}\right),
    \end{equation*}
    where $\bg^{\R}$ and $\bg^{\C}$ are independent real Gaussian random variables each with mean $\mu$ and variance $\sigma^2$. 
\end{definition}

\begin{lemma}[Approximate isometry from Gaussian random matrix]\label{lem:gaussian_concentration}
    Let $\bG$ be an $N\times K$ matrix where each entry is an independent complex Gaussian random variable with mean zero and variance $1/N$. 
    Then with probability $1-2/N^{(\log N)^2}$ over the randomness in $\bG$,
    the matrix $\bX = \sum_{i=1}^m E_i \bG \otimes \bra{i}$ is a $\delta$-approximate isometry,
    where 
    \begin{equation*}
        \delta \leq \sqrt{\frac{Km}{N}} + C \cdot \sqrt{\frac{m \cdot (\log N)^3 }{N}}, 
    \end{equation*}
    for some universal constant $C > 0$. 
\end{lemma}

\begin{proof}
    Let us write
    \begin{equation*}
        \bG = \sum_{a=1}^N \sum_{b=1}^K \frac{1}{\sqrt{N}}\cdot \bg_{ab} \cdot \ketbra{a}{b} = \sum_{a=1}^N \sum_{b=1}^K \frac{1}{\sqrt{2N}}\cdot \left(\bg^{\R}_{ab} + \bg^{\C}_{ab} \cdot \mathrm{i}\right) \cdot \ketbra{a}{b},
    \end{equation*}
    where each $\bg_{ab} = \frac{1}{\sqrt{2}}\cdot (\bg^{\R}_{ab} + \bg^{\C}_{ab}\cdot \mathrm{i})$ is an independent complex Gaussian random variable with zero mean and unit variance. Then
    \begin{align}\label{eq:gaussian_X_matrix}
        \bX \coloneqq \sum_{i=1}^m E_i \bG\otimes \bra{i} &= \frac{1}{\sqrt{2N}}\cdot \sum_{i=1}^m E_i \cdot \left(\sum_{a=1}^N \sum_{b=1}^K \left(\bg^{\R}_{ab} + \bg^{\C}_{ab} \cdot \mathrm{i}\right) \cdot \ketbra{a}{b}\right)\otimes \bra{i} \nonumber\\
        &= \frac{1}{\sqrt{2N}}\cdot \sum_{a=1}^N\sum_{b=1}^K (\bg^{\R}_{ab} + \bg^{\C}_{ab}\cdot \mathrm{i}) \cdot \sum_{i=1}^m  E_i \ketbra{a}{b}\otimes \bra{i} \\
        &= \sum_{s \in \{\R,\C\}} \sum_{a=1}^N\sum_{b=1}^K \bg^s_{ab}\cdot A^s_{ab}, \nonumber
    \end{align}
    where $A^{\R}_{ab} \coloneqq \frac{1}{\sqrt{2N}}\sum_{i=1}^m  E_i \ketbra{a}{b}\otimes \bra{i}$ and $A^{\C}_{ab} \coloneqq A^{\R}_{ab}\cdot \mathrm{i}$. 
    To apply \Cref{thm:BBvH23_matrix_concentration}, we need to calculate
    \begin{align*}
        M_1\coloneqq \sum_{s \in \{\R, \C\}} \sum_{a=1}^N\sum_{b=1}^K (A^s_{ab})^\dagger \cdot A^s_{ab} \qquad \text{and} \qquad 
        M_2\coloneqq\sum_{s \in \{\R, \C\}} \sum_{a=1}^N\sum_{b=1}^K A^s_{ab} \cdot (A^s_{ab})^\dagger. 
    \end{align*}
    First,
    \begin{align*}
        M_1 &= \sum_{a=1}^N\sum_{b=1}^K (A^{\R}_{ab})^\dagger \cdot A^{\R}_{ab} + (A^{\C}_{ab})^\dagger \cdot A^{\C}_{ab} \\
        &= 2\cdot \sum_{a=1}^N\sum_{b=1}^K (A^{\R}_{ab})^\dagger \cdot A^{\R}_{ab} \\
        &= \frac{1}{N}\cdot \sum_{a=1}^N\sum_{b=1}^K \left(\sum_{i_1=1}^m  E_{i_1} \ketbra{a}{b}\otimes \bra{i_1}\right)^\dagger \cdot \sum_{i_2=1}^m  E_{i_2} \ketbra{a}{b}\otimes \bra{i_2} \\
        &= \frac1N\cdot \sum_{a=1}^N\sum_{b=1}^K \sum_{i_1,i_2=1}^m \ketbra{b}{a}E_{i_1}^\dagger E_{i_2} \ketbra{a}{b} \otimes \ketbra{i_1}{i_2} \\
        &= \frac1N\cdot \sum_{i_1,i_2=1}^m \left(\sum_{a=1}^N \bra{a}E_{i_1}^\dagger E_{i_2} \ket{a}\right)\cdot \left(\sum_{b=1}^K \ketbra{b}\right)
         \otimes \ketbra{i_1}{i_2} \\
        &= \frac1N\cdot \sum_{i_1,i_2=1}^m \tr(E_{i_1}^\dagger E_{i_2}) \cdot I_K\otimes \ketbra{i_1}{i_2} \\
        &= I_{Km}. \tag{$\{E_1, \ldots, E_m\}$ is an orthonormal set of errors}
    \end{align*}
    Second,
    \begin{align*}
        M_2 &= 2\cdot \sum_{a=1}^N\sum_{b=1}^K A^{\R}_{ab}\cdot (A^{\R}_{ab})^\dagger \\
        &= \frac{1}{N}\cdot \sum_{a=1}^N\sum_{b=1}^K \sum_{i_1=1}^m  E_{i_1} \ketbra{a}{b}\otimes \bra{i_1} \cdot \left(\sum_{i_2=1}^m  E_{i_2} \ketbra{a}{b}\otimes \bra{i_2} \right)^\dagger \\
        &= \frac1N\cdot \sum_{a=1}^N\sum_{b=1}^K \sum_{i=1}^m E_i\ketbra{a}{a}E_i^\dagger \\
        &= \frac{K}{N}\cdot \sum_{i=1}^m E_i \cdot \left(\sum_{a=1}^N \ketbra{a} \right)\cdot E_i^\dagger \\
        &= \frac{Km}{N} \cdot I_N \tag{$E_i$'s are unitary so $E_i E_i^\dagger = I_N$}. 
    \end{align*}
    Since we assumed that $Km\leq N$, we have that $\sigma(\bX) = (\max \{\norm{M_1}, \norm{M_2}\})^{1/2} =  1$. 
    We now calculate $v(\bX)$. 
    It follows from \Cref{eq:gaussian_X_matrix} that a vectorization of $\bX$ is given by 
    \begin{align*}
        \ket{\bX} \coloneqq \frac{1}{\sqrt{N}}\cdot \sum_{a=1}^N\sum_{b=1}^K \bg_{ab} \cdot \sum_{i=1}^m  E_i \ket{a} \otimes \ket{b}\otimes \ket{i}. 
    \end{align*}
    Hence,
    \begin{align*}
        \E \ketbra{\bX} &= \frac{1}{N}\cdot \sum_{a_1,a_2=1}^N \sum_{b_1,b_2=1}^K \E \left[\bg_{a_1b_1} \cdot \bg^*_{a_2b_2} \right]\cdot \sum_{i_1,i_2=1}^m E_{i_1}\ketbra{a_1}{a_2}E_{i_2}^\dagger \otimes \ketbra{b_1}{b_2} \otimes \ketbra{i_1}{i_2} \\
        &= \frac{1}{N}\cdot \sum_{a=1}^N \sum_{b=1}^K \E \left[\abs{\bg_{ab}}^2\right]\cdot \sum_{i_1,i_2=1}^m E_{i_1}\ketbra{a} E_{i_2}^\dagger \otimes \ketbra{b} \otimes \ketbra{i_1}{i_2} \\
        &= \frac1N\cdot \sum_{a=1}^N \sum_{b=1}^K \sum_{i_1,i_2=1}^m E_{i_1}\ketbra{a} E_{i_2}^\dagger \otimes \ketbra{b} \otimes \ketbra{i_1}{i_2} \\
        &= \frac1N\cdot \sum_{i_1,i_2=1}^m E_{i_1}E_{i_2}^\dagger \otimes I_K \otimes \ketbra{i_1}{i_2}.
    \end{align*}
    Consider the unitary $U = \sum_{i=1}^m E_i \otimes I_K\otimes \ketbra{i}$. Then
    \begin{equation*}
        U^\dagger \cdot \E \ketbra{\bX} \cdot U = \frac1N\cdot \sum_{i_1,i_2=1}^m I_N \otimes I_K\otimes \ketbra{i_1}{i_2} = \frac{m}{N} \cdot I_N \otimes I_K\otimes \ketbra{+^m}{+^m},
    \end{equation*}
    where $\ket{+^m} = \frac{1}{\sqrt{m}} \sum_{i=1}^m \ket{i}$. So
    \begin{equation*}
        v(\bX) = \norm{\E \ketbra{\bX}}^{1/2} = \norm{U^\dagger \cdot \E \ketbra{\bX} \cdot U}^{1/2} = \sqrt{mN^{-1}}. 
    \end{equation*}
    Next, we apply \Cref{thm:BBvH23_matrix_concentration}, which says that there exists a universal constant $C'>0$ such that for any $t\geq 0$, with probability at least $1-2\cdot\exp(-t^2)$ over the randomness in $\bX$, 
    \begin{equation*}
        1 - \delta = \sqrt{s_{\min}(M_1)} - \delta \leq s_{\min}(\bX) \leq s_{\max}(\bX)\leq \sqrt{s_{\max}(M_1)} + \delta = 1+\delta,
    \end{equation*}
    where 
    \begin{align*}
        \delta &= \norm{M_2}^{1/2} + C' \cdot \left( \sigma(\bX)^{1/2}\cdot v(\bX)^{1/2}\cdot (\log N)^{3/4} +  v(\bX)\cdot t \right) \\
        &= \sqrt{\frac{Km}{N}} + C' \cdot \left(\left(\frac{m}{N}\right)^{1/4} \cdot (\log N)^{3/4} +  \left(\frac{m}{N}\right)^{1/2} \cdot t \right). 
    \end{align*}
    Set $t = (\log N)^{3/2}$. Then we have that with probability at least $1 - 2 \cdot \exp(- (\log N)^3) = 1 - \frac{2}{N^{(\log N)^2}}$, $\bX$ is a $\delta$-approximate isometry with
    \begin{align*}
        \delta &\leq \sqrt{\frac{Km}{N}} + C' \cdot \left(\left(\frac{m}{N}\right)^{1/4} \cdot (\log N)^{3/4} +  \left(\frac{m}{N}\right)^{1/2} \cdot (\log N)^{3/2} \right)\\
        &\leq \sqrt{\frac{Km}{N}} + 2 C' \cdot \sqrt{\frac{m \cdot (\log N)^3 }{N} }.
    \end{align*}
    Setting $C = 2C'$ completes the proof. 
\end{proof}

Before we state our the next lemma, we will define the $\mathrm{isometrize}$ operation. 

\begin{definition}[Isometrize]
Let $T$ be a rank-$\ell$ matrix of size $N \times \ell$ where $N \geq \ell$. Write the SVD of $T$ as $T = W \cdot \Sigma \cdot U$, where $W$ is an $N \times \ell$ isometry, $\Sigma$ is an $\ell \times \ell$ diagonal matrix of the singular values, and $U$ is an $\ell \times \ell$ unitary. Then $\mathrm{isometrize}(T) \coloneqq W \cdot U$.\footnote{Note that since $T$ is rank $\ell$, $\mathrm{isometrize}(T)$ is uniquely defined. In particular, $\mathrm{isometrize}(T) = T \cdot (\sqrt{T^\dagger \cdot T})^{-1}$}
\end{definition}

\begin{lemma}
\label{lemma:isometrize}
    Fix a matrix $G \in \mathbb{C}^{N \times K}$ and let $V = \mathrm{isometrize}(G)$. Suppose $X = \sum_{i \in [m]} E_i G \otimes \bra{i}$ is a $\delta$-approximate isometry for $0 \leq \delta < 1$. Then $Y = \sum_{i \in [m]} E_i V \otimes \bra{i}$ is a $2\delta/(1-\delta)$-approximate isometry.
\end{lemma}

\begin{proof}
    First, we will show that $G$ is a $\delta$-approximate isometry. Since $X$ is a $\delta$-approximate isometry, we have that for all unit vectors in $\ket{v} \in \mathbb{C}^{K} \otimes \mathbb{C}^{m}$, 
    \begin{align}
        1-\delta \leq \norm{X \ket{v}} \leq 1+\delta. \label{eq:X-delta-approx}
    \end{align}
    For any unit vector $\ket{c} \in \mathbb{C}^K$, $\norm{X (\ket{c} \ket{1})} = \norm{E_1 G \ket{c}} = \norm{G \ket{c}}$, and thus by \Cref{eq:X-delta-approx}, we have that
    \begin{equation*}
        1-\delta \leq \norm{G \ket{c}} \leq 1+\delta,
    \end{equation*}
    i.e.\ $G$ is a $\delta$-approximate isometry. 

    Next, write the SVD of $G \in \mathbb{C}^{N \times K}$ as $G = W \cdot \Sigma \cdot U$, where $W \in \mathbb{C}^{N \times K}$ is an isometry, $\Sigma \in \mathbb{C}^{K \times K}$ is the diagonal matrix of singular values, and $U \in \mathbb{C}^{K \times K}$ is a unitary. Since $G$ is a $\delta$-approximate isometry, the diagonal entries of $\Sigma$ are between $1-\delta$ and $1+\delta$. Moreover, since $\delta < 1$, $\Sigma$ is invertible. Define $R \coloneqq U^\dagger \cdot \Sigma^{-1} \cdot U$, and note that $V = \mathrm{isometrize}(G) = W \cdot U$ can be written as $V = G \cdot R$.

    Plugging $V = G \cdot R$ into the definition of $Y$, we see that $Y = X \cdot (R \otimes I_m)$. Since $s_{\min}(R) \geq \frac{1}{1+\delta}$ and $s_{\max}(R) \leq \frac{1}{1-\delta}$, we have
    \begin{align*}
        s_{\min}(Y) &\geq s_{\min}(X) \cdot s_{\min}(R) \geq \frac{1-\delta}{1+\delta} \geq 1 - \frac{2\delta}{1-\delta},\\
        s_{\max}(Y) &\leq s_{\max}(X) \cdot s_{\max}(R) \leq \frac{1+\delta}{1-\delta} = 1 + \frac{2\delta}{1-\delta},
    \end{align*}
    which completes the proof.
\end{proof}

\subsection{Proof of \texorpdfstring{\Cref{thm:haar}}{Theorem 1.11}}

\begin{lemma}
\label{lemma:G-dist}
    Let $\bG$ be an $N \times K$ matrix in which each entry $\bG_{i, j}$ is an independent complex Gaussian with mean 0 and variance $1/N$.
    Then $\bG$ is distributed as $\bW \cdot \bSigma \cdot \bU$, where $\bW$ is an independent $N \times K$ Haar random isometry,
    $\bU$ is an independent $K \times K$ Haar random unitary,
    and $\bSigma$ is an independent random $K \times K$ diagonal matrix which is nonsingular with probability 1.
\end{lemma}
\begin{proof}
    The normalized ensemble of $\sqrt{N}\cdot \bG$ is known as the Ginibre ensemble, where each entry is an independent complex Gaussian random variable with mean $0$ and variance $1$.
    It is well-known that the measure of the Ginibre ensemble is invariant under left and right multiplications by arbitrary unitary matrices \cite[Lemma 1]{Mez07}.
    Hence, for any $\bU_1 \in U(N)$ and $\bU_2 \in U(K)$ drawn independently from the Haar measure, the matrix 
    \begin{equation*}
        \bG' = \bU_1 \cdot \bG \cdot \bU_2
    \end{equation*}
    is also an $N\times K$ random Gaussian matrix with the same distribution as $\bG$. 

    Let $\bG = \bA \cdot \boldsymbol{\Sigma} \cdot \bB$ be the singular value decomposition of $\bG$, where $\bA$ is an $N \times K$ isometry, $\boldsymbol{\Sigma}$ is a $K \times K$ diagonal matrix with nonnegative singular values, and $\bB$ is a $K \times K$ unitary. 
    Whenever a singular value is degenerate, the corresponding singular vectors are chosen Haar randomly within the degenerate subspace.
    The singular value decomposition of $\bG'$ is thus $\bG' = (\bU_1 \bA)\cdot \boldsymbol{\Sigma}\cdot (\bB \bU_2)$.
    Note that even conditioned on the value of $\bSigma$, $\bW\coloneqq \bU_1 \bA$ is  distributed as an independent $N \times K$ Haar random isometry, and $\bU \coloneqq \bB \bU_2$ is distributed as an independent $K \times K$ Haar random unitary.

    Finally, we note that $\bSigma$ being nonsingular is equivalent to $\bG$ having rank less than $K$, which occurs when its $K$ random $N$-dimensional Gaussian columns are linearly dependent. This is an event which almost surely does not happen, which finishes the proof.
\end{proof}

We are now ready to prove our main theorem \Cref{thm:haar}, restated below for convenience. 

\mainhaar*

\begin{proof}
    Let $\bG$ be an $N\times K$ matrix where each entry is an independent complex Gaussian random variable with mean zero and variance $1/N$. By~\Cref{lem:gaussian_concentration}, we know that there exists a constant $C > 0$ such that $\bX = \sum_{i \in [m]} E_i \bG \otimes \bra{i}$ is a $\delta'$-approximate isometry with probability $1-2/N^{(\log N)^2}$, where
    \begin{equation*}
         \delta' \coloneqq \sqrt{\frac{Km}{N}} + C \cdot \sqrt{\frac{m \cdot (\log N)^3 }{N}} = \delta/ 3 < 1/3.
    \end{equation*}
    
    By~\Cref{lemma:G-dist}, $\bG$ is distributed as $\bW \cdot \bSigma \cdot \bU$, where $\bW$ is an independent $N \times K$ Haar random isometry,
    $\bU$ is an independent $K \times K$ Haar random unitary,
    and $\bSigma$ is an independent random $K \times K$ diagonal matrix which is nonsingular with probability 1. This implies that with probability $1$, $\mathrm{isometrize}(\bG) = \bW \cdot \bU$ is well-defined and distributed as a Haar random isometry. 

    Set $\bV = \mathrm{isometrize}(\bG)$. Since $\delta'< 1/3 < 1$, we can apply~\Cref{lemma:isometrize} to say that $\bY = \sum_{i =1}^m E_i \bV \otimes \bra{i}$ is a $2\delta'/(1-\delta')$-approximate isometry with probability $1 - 2/N^{(\log N)^2}$. Note that since $\delta' < 1/3$, we can bound $2\delta' /(1-\delta') < 3\delta' = \delta < 1$. 
    Then by~\Cref{thm:decode}, we have that with probability $1 - 2/N^{(\log N)^2}$, $\bV$ specifies an AQEC that corrects for the errors $\errors = \Span\{E_1,\dots,E_m\}$ with disturbance $\delta$.
\end{proof}

\section*{Acknowledgments} 

We thank Thiago Bergamaschi for helpful discussions.

This work was done while F.M. was a postdoctoral fellow at the Simons Institute for the Theory of Computing, supported by DOE QSA grant FP00010905, NSF QLCI Grant 2016245 and DOE grant DE-SC0024124. 
X.T. is supported by the U.S. Department of Energy, Office of Science, National Quantum Information Science Research Centers, Co-design Center for Quantum Advantage (C2QA) under contract number DE-SC0012704. J.W. is supported by the NSF CAREER award CCF-233971.

\bibliographystyle{alpha}
\bibliography{wright}

\begin{thebibliography}{LNCY97}

\bibitem[BBvH23]{BBV23}
Afonso Bandeira, March Boedihardjo, and Ramon {v}an Handel.
\newblock Matrix concentration inequalities and free probability.
\newblock {\em Inventiones mathematicae}, 234(1):419--487, 2023.

\bibitem[Ber24]{Ber24}
Thiago Bergamaschi.
\newblock Pauli manipulation detection codes and applications to quantum communication over adversarial channels.
\newblock In {\em Proceedings of the 43rd Annual International Cryptology Conference}, pages 404--433, 2024.

\bibitem[BGG24]{BGG24}
Thiago Bergamaschi, Louis Golowich, and Sam Gunn.
\newblock Approaching the quantum singleton bound with approximate error correction.
\newblock In {\em Proceedings of the 56th Annual ACM Symposium on Theory of Computing}, pages 1507--1516, 2024.

\bibitem[CGS05]{CGS05}
Claude Cr{\'e}peau, Daniel Gottesman, and Adam Smith.
\newblock Approximate quantum error-correcting codes and secret sharing schemes.
\newblock In {\em Proceedings of the 24th Annual International Cryptology Conference}, pages 285--301, 2005.

\bibitem[HP07]{HP07}
Patrick Hayden and John Preskill.
\newblock Black holes as mirrors: quantum information in random subsystems.
\newblock {\em Journal of high energy physics}, 2007(09):120, 2007.

\bibitem[Kni96]{Kni96}
Emanuel Knill.
\newblock Non-binary unitary error bases and quantum codes.
\newblock Technical report,arXiv:quant-ph/9608048, 1996.

\bibitem[LNCY97]{LNCY97}
Debbie Leung, Michael Nielsen, Isaac Chuang, and Yoshihisa Yamamoto.
\newblock Approximate quantum error correction can lead to better codes.
\newblock {\em Physical Review A}, 56(4):2567, 1997.

\bibitem[LX09]{LX09}
Zhuo Li and Lijuan Xing.
\newblock No more perfect codes: classification of perfect quantum codes.
\newblock Technical report, arXiv:0907.0049, 2009.

\bibitem[Mez07]{Mez07}
Francesco Mezzadri.
\newblock How to generate random matrices from the classical compact groups.
\newblock {\em Notices of the American Mathematical Society}, 54(5):592--604, 2007.

\bibitem[MV16]{MV16}
Benjamin Musto and Jamie Vicary.
\newblock Quantum {L}atin squares and unitary error bases.
\newblock {\em Quantum Information and Computation}, 16(15\&16):1318--1332, 2016.

\bibitem[MW23]{MW23}
Manideep Mamindlapally and Andreas Winter.
\newblock Singleton bounds for entanglement-assisted classical and quantum error correcting codes.
\newblock {\em IEEE Transactions on Information Theory}, 69(9):5857--5868, 2023.

\bibitem[NY23]{NY23}
Ashwin Nayak and Henry Yuen.
\newblock Rigidity of superdense coding.
\newblock {\em ACM Transactions on Quantum Computing}, 4(4):1--39, 2023.

\bibitem[Sha48]{Sha48}
Claude Shannon.
\newblock A mathematical theory of communication.
\newblock {\em The Bell system technical journal}, 27(3):379--423, 1948.

\bibitem[Sho18]{Sho18}
Peter Shor.
\newblock Lecture on quantum error correction.
\newblock Found at \url{https://windowsontheory.org/2018/12/07/quantum-error-correction/}, 2018.

\bibitem[Smi06]{Smi06}
Graeme Stewart~Baird Smith.
\newblock {\em Upper and lower bounds on quantum codes}.
\newblock PhD thesis, California Institute of Technology, 2006.

\bibitem[Var57]{Var57}
Rom Varshamov.
\newblock Estimate of the number of signals in error correcting codes.
\newblock {\em Proceedings of the USSR Academy of Sciences}, 117:739--741, 1957.

\bibitem[vH17]{Van17}
Ramon van Handel.
\newblock Structured random matrices.
\newblock {\em Convexity and concentration}, pages 107--156, 2017.

\bibitem[VW00]{VW00}
Karl Vollbrecht and Reinhard Werner.
\newblock Why two qubits are special.
\newblock {\em Journal of Mathematical Physics}, 41(10):6772--6782, 2000.

\bibitem[Wer01]{Wer01}
Reinhard Werner.
\newblock All teleportation and dense coding schemes.
\newblock {\em Journal of Physics A: Mathematical and General}, 34(35):7081, 2001.

\end{thebibliography}

\end{document}